\documentclass[12pt,a4paper]{article}
\usepackage{fancyhdr}
\usepackage{amsmath}
\usepackage{amsthm,amssymb}
\usepackage{hyperref}
\usepackage{graphicx}
\fancypagestyle{firststyle}
{
	\fancyhead[C]{}
}
\title{A Step by Step Mathematical Derivation and Tutorial on Kalman Filters}
\author{Hamed Masnadi-Shirazi \and Alireza Masnadi-Shirazi \and  Mohammad-Amir Dastgheib}
\newtheorem{theorem}{Theorem}[section]
\newtheorem{corollary}{Corollary}[section]
\newtheoremstyle{example}{1cm}{1cm}{\normalfont}{}{\bfseries}{.}{}{}

\begin{document}

\maketitle

\thispagestyle{firststyle}
\begin{abstract}
	We present a step by step mathematical derivation of the Kalman filter using two different approaches. First, we consider the orthogonal projection method  by means of vector-space optimization. Second, we derive the Kalman filter using Bayesian optimal filtering. We provide detailed proofs for both methods and each equation is  expanded in detail.	
\end{abstract}

\section{Introduction}

The Kalman filter, named after Rudolf E. Kalman, is still a highly useful algorithm today despite having been introduced more than 50 years ago. Its success can be attributed to it being an optimal estimator  and its relatively straightforward and easy to implement recursive algorithm with small computational cost \cite{OptimalF}. 

The Kalman filter has been used in various applications such as smoothing noisy data and providing estimates of parameters of interest,  phase-locked loops in radio equipment, smoothing the output from laptop track pads, global positioning system receivers, and many others \cite{Faragher}. 
	
	

	The Kalman filter \cite{kalman}, also known as the Kalman-Bucy filter \cite{Kalman-bucy}, can be summarized as an iterative prediction-correction process. 
	It can also be seen as a time variant Wiener filter \cite{OptimalF} and was originally derived using the orthogonal projection method. 
	The innovations approach \cite{paoulis} was developed in the late 1960s using martingales theory \cite{innov}, \cite{nonlin}. 
	
	In the first part of this article the orthogonal projection method is used to derive the Kalman filter as a minimum mean squared estimator. 
	The derivation is an expansion of the analysis presented in \cite{Vspace}, such that each step of the proof is clearly derived and presented with complete details.
	
	
	The Kalman filter has a  Bayesian interpretation as well \cite{Bayesian}, \cite{Bayesian2} and can be derived within a Bayesian framework as a MAP estimator.
	The second part of this article uses Bayesian optimal filtering to derive the same equations.
	

\pagestyle{plain}
\newpage
\tableofcontents
\newpage
\section{Model of a Random Process}
	Consider that we have a target state vector $x_k\in \mathbb{R}^{n}$, where $k$ is the time index. The target space evolves according to the discrete time stochastic model:
	\[x_k=\phi_{k-1}(x_{k-1},u_{k-1})\]
	$\phi_{k-1}$ is a known , possibly nonlinear function of state $x_{k-1}$ and $u_{k-1}$ is the noise which counts e.g. for mis-modeling or disturbances in target motion.
	
	Also consider that the measurements of the process (picked up by the sensor for example) are $z_k\in\mathbb{R}^m$. the measurements and states are related by
	\[z_k=h_k(x_k,w_k)\]
	where $h_k$ is a known, possibly nonlinear function and $w_k$ is the measurement noise.
	
	$w_k$ and $u_{k-1}$ are assumed to be white with known probability distribution functions and independent of each other.
	
	Filtering is an operation that involves extraction of information about a quantity of interest $x_k$ at (discrete) time $k$ by using data measured up to and including time $k$. Therefore, the objective of filtering is to recursively estimate $x_k$ (target state) from the measurements $z_k$. 
	
	For the special case where $\phi_k$ and $h_k$ are linear functions and the distribution of noise and initial states are Gaussian, the $n$-dimensional dynamic model of a random process reduces to the following linear/Gaussian model and  consists of the following three parts:
	\begin{enumerate}
		\item A vector with difference equation \[x_{k+1}=\Phi_{k}x_{k}+u_{k}\qquad k=0,1,2,\ldots\] which defines how the random vector $x_{k}$ changes with time.
		
		\begin{itemize}
			\item Here $x_{k}$ is an $n$-dimensional state vector where each component is a random variable.
			\item $\Phi_{k}$ is a known $n\times n$  matrix.
			\item $u_{k}$ is an $n$-dim random vector of input with zero mean and there is zero correlation between present noise at the time $k$ and past noise at time $l$, i.e:
			\[ E[u_{k}u_{l}']=Q_{k}\delta_{kl}=\left\lbrace\begin{array}{ll}
			Q_{k}&\qquad k=l\\
			0&\qquad k\neq l
			\end{array}\right. \]
			where $Q_{k}>0$ is a positive definite matrix.
		\end{itemize}
		\item An initial random vector $x_{0}$ and initial random estimate $ \hat{x}_{0} $ with initial error covariance $ E[(x_{0}-\hat{x}_{0})(x_{0}-\hat{x}_{0})']=P_{0}$ 
		\item Measurements of the process is of the form 
		\[ z_{k}=H_{k}x_{k}+w_{k}\qquad k=0,1,2,\ldots \]
		which defines how the measurements $z_{k}$ of the process $x_{k}$ are measured over time.
		\begin{itemize}
			\item Here $H_{k}$ is a known $m\times n$ matrix
			\item $w_{k}$ is an $n$-dimensional random measurement error with zero mean and 
			\[ E[w_{k}w_{l}']=R_{k}\delta_{kl}=\left\lbrace\begin{array}{ll}
			R_{k}&\qquad k=l\\
			0&\qquad k\neq l
			\end{array}\right. \] 
			where $R_{k}>0$ is a positive definite matrix.
		\end{itemize}
		
		It is assumed that $x_{0}$,$u_{j}$,$w_{k}$ are all uncorrelated for $j\geq0,k\geq0$.
	\end{enumerate}

\part{Derivation Using Vector Space Methods}
\section{Hilbert Space of Random Vectors}
\subsection{A Review of Probability}
For a  real valued random variable $x$, we define the probability distribution $P$ of $x$ by
\[ P(\zeta)=\textrm{Prob}(x\leq\zeta). \] 
In other words, $P(\zeta)$ is the probability that the random variable $x$ assumes a value less than or equal to the number $\zeta$. 
For a finite collection of real random variables $\{ x_1, x_2,..., x_n\}$,their joint probability distribution $P$ shows their inter-dependencies
and is defined as 
\[ P(\zeta_1,\zeta_2,\ldots,\zeta_n)=\textrm{Prob}(x_1\leq\zeta_1,x_2\leq\zeta_2,\ldots,x_n\leq\zeta_n). \]

It is often useful to characterize a random variable by its mean and variance. Therefore the following quantities are of primary interest. 
\begin{eqnarray}
&{E[x]} \hspace{40pt} \textrm{ is the \textit{expected value} of } x.\nonumber \\
&{E[x^2]}\hspace{40pt} \textrm{ is the \textit{average power} of } x. \nonumber\\	
&{E[(x-E(x))^2]} \hspace{25pt} \textrm{ is the \textit{variance} of } x.\nonumber
\end{eqnarray} 
Note that the mathematical expectation operator $E[x]$ is a linear operator.

If $g(.)$ is a single valued function then $g(x)$ is also a random variable and its expected value is  defined as 
\[ E[g(x)]=\int_{-\infty}^{\infty} g(\zeta) dP(\zeta), \]
which may in general not be finite. 
Also, the expected value of any function $g(.)$ over a collection of random variables $\{ x_1, x_2,..., x_n\}$ is defined as
\[
E[g(x_1,x_2,\ldots,x_n)]=\int_{-\infty}^{\infty}\cdots\int_{-\infty}^{\infty} g(\zeta_1,\zeta_2,\ldots,\zeta_n) dP(\zeta_1,\zeta_2,\ldots,\zeta_n). 
 \]
 
 The second-order statistical averages of these variables can be 
 described in terms of expected values.
Specifically, for the  $n\times n$ \textit{covariance matrix}  $\textrm{\textit{cov}}(x_1,x_2,\ldots,x_n)$,  its $ij$-th element is 
defined as 
\[ 
E[(x_i-E(x_i))(x_j-E(x_j))]=E(x_ix_j)-E(x_i)E(x_j),
 \]
 which in case of zero means reduces to $E(x_ix_j)$.\\ 
 
Finally, if $E(x_ix_j)=E(x_i)E(x_j)$ then $E[(x_i-E(x_i))(x_j-E(x_j))]=0$ and we say that $x_i$ and $x_j$ are uncorrelated. 

 \subsection{Random vectors}
 The idea of random variables can be generalized to random vectors. An $n$-dimensional random vector $x$  is an ordered set of $n$ random values $x_i$ and is defined as
 \[ x=\left[
 \begin{array}{c}
 x_1 \\
 x_2 \\
 \vdots\\
 x_n
 \end{array}
 \right] \]
 
 Let $\{y_1,y_2,\ldots,y_m\}$ be $n$-dimensional random vectors of the above form then a Hilbert space $\mathcal{H}$ can be defined such that $\mathcal{H}$ consists of all vectors whose components are  linear combination of the $y_i$'s.\\
 If $x$ and $y$ are elements of $\mathcal{H}$, we define their inner product as
 \[ (x|y)=E(x^Ty)=E(\sum_{i=1}^{n}x_iy_i). \]
 The induced norm of a vector $x$ in this space can be written as
 \begin{eqnarray}
 \|x\|&=&\sqrt{E(x^Tx)}=\sqrt{E(x_1^2+x_2^2+\ldots+x_n^2)}\\ 
&=&\sqrt{E(x_1^2)+E(x_2^2)+\ldots+E(x_n^2)} \\
&=&\{\mathrm{Trace}(E(xx^T))\}^\frac{1}{2},
\end{eqnarray}
 since the expected value of the random matrix $xx^T$ is 
 \[
 E(xx^T)= \left[\begin{array}{cccc}
 E(x_1x_1) & E(x_1x_2) & \cdots & E(x_1x_n)\\
 E(x_2x_1) & E(x_2x_2) & \cdots & E(x_1x_n)\\
 \vdots & \vdots& \ddots& \vdots\\
 E(x_nx_1) & E(x_nx_2) & \cdots & E(x_nx_n)
 \end{array}\right].
 \]
 Similarly the inner product can also be written as
\[ (x|y)=\textrm{Trace}(E(xy^T)). \]

Two vectors are said to be orthogonal if $(x|y)=0$ and this can be written as $x\perp y$. If $x$ and $y$ are uncorrelated and $E(x)=E(y)=0$ then $x$ and $y$ are orthogonal to each other since
\[ (x|y)=E(x^Ty)=E(x)E(y)=0 \Rightarrow x\perp y. \]

Finally, the covariance matrix for a random vector is defined as 
\[\textrm{cov}(x)=E[(x-E(x))(x-E(x))^T]. \]
If $E(x)=0$ then  the covariance matrix can be written as
\[ E(xx^T).\]

\section{Minimum Variance Unbiased (Gauss-Markov) Estimate}
\subsection{Problem Setup}
Assume that we have observations of a variable $y=\left[\begin{array}{c}y_1\\y_2\\\vdots\\y_m
\end{array}\right]$ and  that $y$ is a linear estimation of other variables $W_{m\times n}$ plus an error term to account for measurement errors such that
\[ y_i=\beta_1w_{i1}+\beta_2w_{i2}+\ldots+\beta_nw_{in}+\varepsilon_i.\]
We can therefore write
\[ y=W\beta+\varepsilon, \]
where $y$ is the known outcome of $m$ inexact measurements, $W$ is a known matrix, $\beta$ is an unknown vector of parameters and $\epsilon$ is a random vector
such that 
\[ E(\varepsilon_i)=0, \]
\[ E(\varepsilon\varepsilon')=Q>0, \]
where $Q$ is a positive definite matrix.
Assuming $W$ and $y$ are known we want to estimate the unknown $\beta$. We seek a linear estimate $\hat{\beta}=Ky$, where $K_{m\times n}$ is an unknown constant.

Since $\varepsilon$ is a random vector,  $y=W\beta+\varepsilon$ is  a random vector and since $\hat{\beta}=Ky$, $\hat{\beta}$ is also a random vector. As a result the estimation error defined by $\textrm{error}=\hat{\beta}-\beta$ is a random vector as well.

We consider the optimality criterion of minimizing the norm of the error in order to  find $\hat{\beta}$. Since error is a random vector, the norm is defined as
\begin{align}
 \|\textrm{error}\|^2=& E[(\hat{\beta}-\beta)^T(\hat{\beta}-\beta)]=E[(Ky-\beta)^T(Ky-\beta)] \nonumber\\
 =& E[(K(W\beta+\varepsilon)-\beta)^T(K(W\beta+\varepsilon)-\beta)] \nonumber\\
 =& E[(KW\beta+K\varepsilon-\beta)^T(KW\beta+K\varepsilon-\beta)]\nonumber\\
 =& E[((KW\beta)^T+(K\varepsilon)^T-\beta^T)(KW\beta+K\varepsilon-\beta)]  \nonumber
 \end{align}
 after multiplying out
 \begin{align}
 = E[&\underline{(KW\beta)^T(KW\beta)}+( KW\beta)^T(K\varepsilon) -\underline{(KW\beta)^T\beta}\nonumber\\
 +&(K\varepsilon)^T(KW\beta)+( K\varepsilon)^T(K\varepsilon )-(K\varepsilon)^T\beta\nonumber\\
 -&\underline{\beta^T(KW\beta)}-\beta^T(K\varepsilon)+\underline{\beta^T\beta}] \nonumber
 \end{align}
 separating the underlined terms
\begin{align}
= E[&(KW\beta)^T(KW\beta)+\beta^T\beta-(KW\beta)^T\beta-\beta^T(KW\beta)]\nonumber\\
+ E[& \beta^TW^TK^TK\varepsilon+\varepsilon^TK^TKW\beta +(K\varepsilon)^T(K\varepsilon)
-\varepsilon^TK^T\beta-\beta^TK\varepsilon] \nonumber
\end{align}
and moving constant terms out of the expectation
\begin{align}
=&(KW\beta)^T(KW\beta)+\beta^T\beta-(KW\beta)^T\beta-\beta^T(KW\beta)\nonumber\\
+&\beta^TW^TK^TKE[\varepsilon]+E[\varepsilon^T]K^TKW\beta +E[(K\varepsilon)^T(K\varepsilon)]
-E[\varepsilon^T]K^T\beta-\beta^TKE[\varepsilon]. \nonumber
\end{align}
Since the expected value of $\varepsilon$ is zero then $E[\varepsilon]=E[\varepsilon^T]=0$ and
\begin{align}
\|\textrm{error}\|^2 =& \|KW\beta-\beta\|^2+E[(K\varepsilon)^T(K\varepsilon)] \nonumber\\
=& \|KW\beta-\beta\|^2+\textrm{Trace}(E[K\varepsilon(K\varepsilon)^T])\nonumber\\
=& \|KW\beta-\beta\|^2+\textrm{Trace}(E[K\varepsilon\varepsilon^TK^T])\nonumber\\
=& \|KW\beta-\beta\|^2+\textrm{Trace}(KE[\varepsilon\varepsilon^T]K^T)\nonumber\\
=& \|KW\beta-\beta\|^2+\textrm{Trace}(KQK^T).\nonumber
\end{align}

Note that we are trying to find the unknown $K$ such that it minimizes the error.  Yet, the error, in the expression above, is also a function of the unknown $\beta$.
If $KW=I$ then the error expression is independent of $\beta$ since
\[ \|KW\beta-\beta\|=\|I\beta-\beta\|=0. \]
The problem can now be written as
\[ \begin{array}{cccc}
 \arg\min_{\hat{\beta}} & \|\hat{\beta}-\beta\|^2  & =  \arg\min_{K} & \textrm{Trace}(KQK')\\
 &\text{s.t. } KW=I &  &\text{s.t. } KW=I
\end{array}\]
which is independent of $\beta$.

\subsection{What does imposing $KW=I$ mean?}
We define the estimate $\hat{\beta}$ of an operator $\beta$ to be unbiased if $E[\hat{\beta}]=\beta$.
If we impose $KW=I$ then we can write 
\[ E[\hat{\beta}]=e[Ky]=E[KW\beta+K\varepsilon]=E[KW\beta]+E[K\varepsilon]=\underbrace{KW}_I\beta=\beta. \]
Therefore imposing $KW=I$ is equivalent to requiring that $\hat{\beta}$ be an unbiased estimate of $\beta$. 
In summary, we are trying to find the unbiased linear estimate of $\beta$ that minimizes $\|\hat{\beta}-\beta\|^2$.

\subsection{Solution to the problem}
The minimization problem above can be written in terms of the elements of $\hat{\beta}$ as
\[ \begin{array}{cl}
\arg\min_{\hat{\beta}} & \sum_{i=1}^{n}E[(\hat{\beta}_i-\beta_i)^2]\\
\text{s.t. } & E[\hat{\beta}_i]=\beta_i \, ,i=1,2,\ldots,n\\
& \hat{\beta}_i=k_i'y \,  ,i=1,2,\ldots,n  
\end{array}\]
Where $k_i'$ is the $i$th row of the matrix $K$.

Since $E[(\hat{\beta}_i-\beta_i)^2]>0$, every term in the above summation is nonnegative, therefore the sum $\sum_{i=1}^{n}E[(\hat{\beta}_i-\beta_i)^2]$ is minimum when each term ,$E[(\hat{\beta}_i-\beta_i)^2] $, is minimized. So we can solve $n$ separate problems, one for each $\beta_i$ as

\[ \begin{array}{cl}
\arg\min_{\hat{\beta_i}} & E[(\hat{\beta}_i-\beta_i)^2]\\
\text{s.t. } & E[\hat{\beta}_i]=\beta_i \, ,i=1,2,\ldots,n\\
& \hat{\beta}_i=k_i^Ty \,  ,i=1,2,\ldots,n  
\end{array}\]

We can also write the problem as finding the optimal matrix $K$
\[ \begin{array}{cl}
\arg\min_{K} & \textrm{Trace}(KQK^T)\\
\text{s.t. } & KW=I
\end{array}\]
This can be thought of as a minimum weighted norm problem in the space of matrices, or it can  also be decomposed into $n$ separate problems where the $i$th problem is 
\[ \begin{array}{cl}
\arg\min_{k_i} & k_i^TQk_i\\
\text{s.t. } & k_i^Tw_j=\delta_{ij} \,,\, i,j=1,2,\ldots,n
\end{array}\]
 where $w_j$ is the $j$th column of $W$, $k_i$ is the $i$th row of $K$ and $\delta_{ij}$ is the Kronecker delta function defined as
\[\delta_{ij}= \left\lbrace\begin{array}{cl}
0 & i\neq j\\
1 & i=j
\end{array}\right.\]
Defining the weighted inner product as $(x|y)_Q=x^TQy$ and noting that $ (k_i|Q^{-1}w_j)_Q=k_i^TQQ^{-1}w_j=k_i^Tw_j$ the above problem can be written as
\[ \begin{array}{cl}
\arg\min_{k_i} & (k_i|k_i)_Q\\
\text{s.t. } & (k_i|Q^{-1}w_j)_Q=\delta_{ij}
\end{array}\]

This is in the form of the standard minimum norm problem and can be rewritten as 
 \[ \begin{array}{clcll}
\arg\min_{k_i} & \|k_i\|^2_Q &\equiv &\arg\min_{k_i} & \|k_i\|^2_Q\\
\text{s.t. } & \begin{array}{c}
k_i^Tw_1=0 \\
\vdots\\
k_i^Tw_i=1\\
\vdots\\
k_i^Tw_n=0
\end{array}
& &\text{s.t. } &\left[\begin{array}{c}\textendash w_1\textendash \\
\vdots\\
\textendash w_i\textendash \\
\vdots\\
\textendash w_n\textendash \end{array}\right]
\left[\begin{array}{c}| \\
k_i\\
|\end{array}\right]=
\left[\begin{array}{c} 0 \\
\vdots\\
1\\
\vdots\\
0 \end{array}\right]=e_i
\end{array}\] 

The above problem can be summarized as 
\[ \begin{array}{cl}
\arg\min_{k_i} &  \|k_i\|_Q^2 \\
\text{s.t. } & W^Tki=e_i
\end{array}\]
Assuming that $W$ is full column rank then  $W^T$ is full row rank and the least squares solution  is 
\[ k_i=Q^{-1}W(W^TQ^{-1}W)^{-1}e_i. \]
We can now find  $K^T$ by combining all the $k_i$'s as
\[ K^T=Q^{-1}W(W^TQ^{-1}W)^{-1} \]
and write the final solution $\hat{\beta}$ as
\[  \hat{\beta}=Ky=(W^TQ^{-1}W)^{-1}W^TQ^{-1}y. \] 
Here, we also compute the error covariance matrix  as
\begin{align}
E[(\hat{\beta}-\beta)(\hat{\beta}-\beta)^T]=& E[(Ky-\beta)(Ky-\beta)^T] = E[(Ky-\beta)(y^TK^T-\beta^T)]\nonumber\\
=& E[(KW\beta+K\varepsilon-\beta)(\beta^TW^TK^T+\varepsilon^TK^T-\beta^T)\nonumber\\
=& E[KW\beta\beta^TW^TK^T+KW\beta\varepsilon^TK^T-KW\beta\beta^T+K\varepsilon\beta^TW^TK^T\nonumber\\
+&K\varepsilon\varepsilon^TK^T-K\varepsilon\beta^T-\beta\beta^TW^TK^T-\beta\varepsilon^TK^T+\beta\beta^T]\nonumber\\
\stackrel{E[\varepsilon]=0}{=}& E[(KW\beta-\beta)(KW\beta-\beta)^T]+K\underbrace{E[\varepsilon\varepsilon^T]}_{Q}K^T\nonumber
\end{align}
 inserting $K=(W^TQ^{-1}W)^{-1}W^TQ^{-1}$ 
\[\begin{array}{l}
E[(\hat{\beta}-\beta)(\hat{\beta}-\beta)^T]=\\
=E[(\underbrace{(W^TQ^{-1}W)^{-1}}_{A^{-1}}\underbrace{W^TQ^{-1}W}_{A}\beta-\beta)(\underbrace{(W^TQ^{-1}W)^{-1}}_{A^{-1}}\underbrace{W^TQ^{-1}W}_{A}\beta-\beta)^T]+KQK^T\\
=\underbrace{E[(\beta-\beta)(\beta-\beta)^T]}_0+(W^TQ^{-1}W)^{-1}\underbrace{W^TQ^{-1}\underbrace{QQ^{-1}}_IW}_A\underbrace{(W^TQ^{-1}W)^{-1}}_{A^{-1}}\\
=(W^TQ^{-1}W)^{-1}
\end{array} \] 

\section{Minimum Variance Estimate}
In the previous discussion $\beta$ was assumed to be unknown and could take any value from $-\infty$ to $+\infty$. We had no prior knowledge about its values. If we have prior knowledge, such as $\beta$'s mean or covariance, then this prior info can be used to produce an estimate with lower error variance compared to the minimum variance unbiased estimate.

So, we assume that $y=W\beta+\varepsilon$ but in this case both $\varepsilon$ and $\beta$ are random vectors. We again want to find $\hat{\beta}$ such that we minimize the norm of the error.

\subsection{Minimum Variance Estimate Theorem}
\begin{theorem}\label{thm:Minimum Variance}
	Let $y$ and $\beta$ be random vectors. Assume that $[E[yy^T]]^{-1}$ exists. The linear estimate $\hat{\beta}$ of $\beta$ based on $y$ that minimizes $\|\hat{\beta}-\beta\|^2$ is
	\[ \hat{\beta}=E[\beta y^T][E[yy^T]]^{-1}y, \]
	with corresponding error covariance matrix
	\[ E[(\hat{\beta}-\beta)(\hat{\beta}-\beta)^T]=E[\beta\beta^T]-E[\hat{\beta}\hat{\beta}^T]=E[\beta\beta^T]-E[\beta y^T][E[yy^T]]^{-1}E[y\hat{\beta}].\]
\end{theorem}
\begin{proof}
	Similar to the previous problem, this problem decomposes into a separate problem for each $\beta_i$. There are no constraints so we find the best approximation 
	of  $\beta_i$ within the subspace generated by the $y_i$s.
	
	Writing the optimal estimate as $\hat{\beta}=Ky$ where $K_{m\times n}$, then the $i$th sub-problem  is equivalent to the problem of selecting the $i$th row of $K$, which in turn gives the optimal linear combination of $y_i$s that make $\beta_i$s.
	So each row of $K$ should satisfy the normal equations corresponding to projecting $\beta_i$ onto the $y_i$. Specifically, 
\[ Ky=\left[\begin{array}{c}
\textendash k_1^T \textendash\\
\vdots\\
\textendash k_i^T \textendash\\
\vdots\\
\textendash k_n^T \textendash\\
\end{array}\right]\left[\begin{array}{c}
y_1\\
\vdots\\
y_i\\
\vdots\\
y_n\\
\end{array}\right] 
=\left[\begin{array}{c}
\beta_1\\
\vdots\\
\beta_i\\
\vdots\\
\beta_n\\
\end{array}\right]=\beta\]
and  $\textrm{error}_i=\beta_i-\hat{\beta}_i=\beta_i-k_i^Ty$ should be orthogonal to each $y_j$ (orthogonality principle), consequently
\[
(\beta_i-k_i^Ty|y_j)=0\Rightarrow(\beta_i|y_j)-(k_i^Ty|y_j)=0\Rightarrow(k_i^Ty|y_j)=(\beta_i|y_j)
\]
\[ \Rightarrow\left\lbrace\begin{array}{c}
(k_i^Ty|y_1)=(\beta_i|y_1)\\
(k_i^Ty|y_2)=(\beta_i|y_2)\\
\vdots\\
(k_i^Ty|y_n)=(\beta_i|y_n)
\end{array}\right.\Rightarrow\left\lbrace\begin{array}{c}
(k_{i1}y_1+k_{i2}y_2+\ldots+k_{in}y_n|y_1)=(\beta_i|y_1)\\
(k_{i1}y_2+k_{i2}y_2+\ldots+k_{in}y_n|y_2)=(\beta_i|y_2)\\
\vdots\\
(k_{i1}y_1+k_{i2}y_2+\ldots+k_{in}y_n|y_n)=(\beta_i|y_n)
\end{array}\right. \]
\begin{alignat}{2}
 \Rightarrow&\left\lbrace\begin{array}{c}
k_{i1}(y_1|y_1)+k_{i2}(y_2|y_1)+\ldots+k_{in}(y_n|y_1)=(\beta_i|y_1)\\
k_{i1}(y_1|y_2)+k_{i2}(y_2|y_2)+\ldots+k_{in}(y_n|y_2)=(\beta_i|y_2)\\
\vdots\\
k_{i1}(y_1|y_n)+k_{i2}(y_2|y_n)+\ldots+k_{in}(y_n|y_n)=(\beta_i|y_n)
\end{array}\right. \nonumber\\ 
 \Rightarrow&\left\lbrace\begin{array}{c}
k_{i1}E[y_1y_1]+k_{i2}E[y_2y_1]+\ldots+k_{in}E[y_ny_1]=E[\beta_i|y_1]\\
k_{i1}E[y_1y_2]+k_{i2}E[y_2y_2]+\ldots+k_{in}E[y_ny_2]=E[\beta_i|y_2]\\
\vdots\\
k_{i1}E[y_1y_n]+k_{i2}E[y_2y_n]+\ldots+k_{in}E[y_ny_n]=E[\beta_i|y_n]
\end{array}\right. \nonumber
 \end{alignat} 
\[ \Rightarrow\underbrace{\left[\begin{array}{cccc}
E[y_1y_1]& E[y_2y_1]&\ldots& E[y_ny_1]\\
E[y_1y_2]& E[y_2y_2]&\ldots& E[y_ny_2]\\
\vdots & \ddots & & \vdots\\
E[y_1y_n]& E[y_2y_n]&\ldots& E[y_ny_n]
\end{array}\right]}_{\text{symmetric}}\left[\begin{array}{c}
k_{i1}\\
\vdots\\
k_{in}\\
\end{array}\right] =\left[\begin{array}{c}
E[\beta_iy_1]\\
\vdots\\
E[\beta_iy_n]\\
\end{array}\right]  \] 
\[ \Rightarrow\underbrace{\left[\begin{array}{cccc}
	E[y_1y_1]& \ldots& E[y_1y_n]\\
	E[y_2y_1]& \ldots& E[y_ny_2]\\
	\vdots & \ddots &  \vdots\\
	E[y_ny_1]& \ldots& E[y_ny_n]
	\end{array}\right]}_{E[yy^T]}\underbrace{\left[\begin{array}{c}
k_{i1}\\
\vdots\\
k_{in}\\
\end{array}\right]}_{k_i^T} =\left[\begin{array}{c}
E[\beta_iy_1]\\
\vdots\\
E[\beta_iy_n]\\
\end{array}\right]  \] 
we have these matrix equations for every $i$, which can all be combined and written as
\[ E[yy^T]K^T=E[y\beta^T]\Rightarrow K^T=[E[yy^T]]^{-1}E[y\beta^T] \]
\[ K=E[\beta y^T][E[yy^T]]^{-1}. \]
	
The error covariance matrix can now be written as
	\begin{align}
		E[(\hat{\beta}-\beta)(\hat{\beta}-\beta)^T]=& E[(Ky-\beta)(Ky-\beta)^T] = E[(Ky-\beta)(y^TK^T-\beta^T)]\nonumber\\
		=& E[Kyy^TK^T-Ky\beta^T-\beta y^TK^T+\beta\beta^T]\nonumber\\
		=& KE[yy^T]K^T-KE[y\beta^T]-E[\beta y^T]K^T+E[\beta\beta^T].\nonumber
	\end{align}
	Noting that
	\begin{align}
		K=E[\beta y^T][E[yy^T]]^{-1}\nonumber\\
		K^T=[E[yy^T]]^{-1}E[y \beta^T]\nonumber
	\end{align}
	and substituting  for $K$ and $K^T$ we find
	\begin{align}
	E[(\hat{\beta}-\beta)(\hat{\beta}-\beta)^T]=& E[\beta y^T][E[yy^T]]^{-1}\underbrace{E[yy^T][E[yy^T]]^{-1}}_IE[y \beta^T]\nonumber\\
	-& E[\beta y^T][E[yy^T]]^{-1}E[y\beta^T]-E[\beta y^T][E[yy^T]]^{-1}E[y \beta^T]+E[\beta\beta^T]\nonumber\\
	=& E[\beta\beta^T]-E[\beta y^T][E[yy^T]]^{-1}E[y \beta^T]\nonumber
	\end{align}
	\end{proof}


If $\beta$ and $y$ have zero mean then 
\[ E[y]=0=E[W\beta+\varepsilon]=WE[\beta]+E[\varepsilon]\stackrel{E[\beta]=0}{=\joinrel=}E[\varepsilon]=0 \]
and we can write
\[ E[\hat{\beta}]=E[Ky]=E[KW\beta+K\varepsilon]=KWE[\beta]+KE[\varepsilon]=0=E[\beta]. \]
Therefore, $\hat{\beta}$ is an unbiased estimate of $\beta$.

Also, note that $\|\hat{\beta}-\beta\|^2$ can be written as
\[ \|\hat{\beta}-\beta\|^2 = E[(\hat{\beta}-\beta)^T(\hat{\beta}-\beta)] = E[\|(\hat{\beta}-\beta)\|^2_2] \]
where $\|.\|_2$ is the standard two-norm and we denote $E[\|(\hat{\beta}-\beta)\|^2_2]$ as the error variance. 

\begin{corollary}\label{col:1}
	Suppose that $y=W\beta+\varepsilon$, where $y$ is a known $m$-dimensional vector, $\beta$ is an $n$-dimensional unknown random vector, $\varepsilon$ is an 
	unknown $m$-dimensional random vector and $W_{m\times n}$ is a known  constant matrix and
	\begin{align}
	& E[\varepsilon\varepsilon^T]=Q\geq0&\text{ (noise covariance)}\nonumber\\
	& E[\beta\beta^T]=R\geq0&\text{ (input covariance for $\beta$)}\nonumber\\
	& E[\varepsilon\beta^T]=0&\text{ (no correlation between input and noise)}\nonumber
	\end{align}
	we also assume that $WRW^T+Q$ is invertible.
	
	Then the linear estimate $\hat{\beta}$ of $\beta$ that minimizes the error variance \mbox{$E[\|\hat{\beta}-\beta\|^2_2]$} is 
	\begin{equation}
	\hat{\beta}=RW^T(WRW^T+Q)^{-1}y
	\end{equation}
	with error covariance
	\[ E[(\beta-\hat{\beta})(\beta-\hat{\beta})^T]=R-RW^T(WRW^T+Q)^{-1}WR \]
\end{corollary}
\begin{proof}
	\begin{align}
	E[yy^T]&=E[(W\beta+\varepsilon)(W\beta+\varepsilon)^T]\nonumber\\
	&=E[W\beta\beta^TW^T+W\beta\varepsilon^T+\varepsilon\beta^TW^T+\varepsilon\varepsilon^T]\nonumber\\
	&=W\underbrace{E[\beta\beta^T]}_{R}W^T+W\underbrace{E[\beta\varepsilon^T]}_0+\underbrace{E[\varepsilon\beta^T]}_0W^T+\underbrace{E[\varepsilon\varepsilon^T]}_Q\nonumber\\
	&=WRW^T+Q\nonumber
	\end{align}
	and
	\begin{align}
	E[\beta y^T]&=E[\beta(W\beta+\varepsilon)^T]=E[\beta(\beta^TW^T+\varepsilon^T)]\nonumber\\
	&=E[\beta\beta^TW^T+\beta\varepsilon^T]=\underbrace{E[\beta\beta^T]}_{R}W^T+E[\beta\varepsilon^T]=RW^T,\nonumber
	\end{align}therefore
\[ \hat{\beta}=E[\beta y^T][E[yy^T]]^{-1} = RW^T(WRW^T+Q)^{-1}y\]
 and
\begin{align}
E[(\beta-\hat{\beta})(\beta-\hat{\beta})]&=
E[\beta\beta^T]-E[\beta y^T][E[yy^T]]^{-1}E[y\beta^T]\nonumber\\
&=R-RW^T(WRW^T+Q)^{-1}WR^T\nonumber\\
&=R-RW^T(WRW^T+Q)^{-1}WR,\nonumber
\end{align}  
since $E[y\beta^T]=WR^T=WR$.
\end{proof}
\begin{corollary}
	The estimate given by corollary \ref{col:1} can be expressed in the alternative form 
	\begin{equation}
	\label{eq:BetaHatAlt}
	\hat{\beta}=(W^TQ^{-1}W+R^{-1})^{-1}W^TQ^{-1}y 
	\end{equation}
	with corresponding error covariance 
	\[ E[(\hat{\beta}-\beta)(\hat{\beta}-\beta)^T]=(W^TQ^{-1}W+R^{-1})^{-1}.\]
\end{corollary}\label{col:2}
\begin{proof}
	We need to show that $RW^T(WRW^T+Q)^{-1}=(W^TQ^{-1}W+R^{-1})^{-1}W^TQ^{-1}$. We  prove this by pre-multiplying both sides by $(WRW^T+Q)$ and post-multiplying  both sides by $(W^TQ^{-1}W+R^{-1})$.
	\begin{align}
	&(W^TQ^{-1}W+R^{-1})\big[RW^T(WRW^T+Q)^{-1}\big](WRW^T+Q)=\nonumber\\
	&(W^TQ^{-1}W+R^{-1})\big[(W^TQ^{-1}W+R^{-1})^{-1}W^TQ^{-1}\big](WRW^T+Q)\nonumber\\
	\Leftrightarrow & (W^TQ^{-1}W+R^{-1})RW^T=W^TQ^{-1}(WRW^T+Q)\nonumber\\
	\Leftrightarrow & W^TQ^{-1}WRW^T+\underbrace{R^{-1}R}W^T_I=W^TQ^{-1}WRW^T+W^T\underbrace{Q^{-1}Q}_I \nonumber\\
	\Leftrightarrow & W^TQ^{-1}WRW +W^T=W^TQ^{-1}WRW^T+W^T\,\checkmark\nonumber
	\end{align}
	
	Substituting in corollary \ref*{col:1} we have
	\begin{align}
	E[(\hat{\beta}-\beta)(\hat{\beta}-\beta)^T]&=R-RW^T(WRW^T+Q)^{-1}WR\nonumber\\
	&=R-(W^TQ^{-1}W+R^{-1})^{-1}W^TQ^{-1}WR\nonumber\\
	&=(W^TQ^{-1}W+R^{-1})^{-1}(W^TQ^{-1}W+R^{-1})R\nonumber\\
	&-(W^TQ^{-1}W+R^{-1})^{-1}W^TQ^{-1}WR\nonumber\\
	&=(W^TQ^{-1}W+R^{-1})^{-1}\big[(W^TQ^{-1}W+R^{-1})R-W^TQ^{-1}WR\big]\nonumber\\
	&=(W^TQ^{-1}W+R^{-1})^{-1}\big[W^TQ^{-1}WR+R^{-1}R-W^TQ^{-1}WR\big]\nonumber\\
	&=(W^TQ^{-1}W+R^{-1})^{-1}\,\checkmark\nonumber
	\end{align}
\end{proof}
If we compare equation \ref{eq:BetaHatAlt} of corollary \ref*{col:2} to the Gauss-Markov estimate we see that if $R^{-1}=0$, corresponding to infinite variance of prior  on $\beta$, then the minimum-variance estimate is equal to the Gauss-Markov estimate. In other words, the Gauss-Markov estimate is a special case of the minimum-variance estimate, when we have no prior information on $\beta$.

\subsection{Preliminary Theorems}
\begin{theorem}\label{th:Tbeta}
	The minimum variance linear estimate of a linear function  of $\beta$ is equal to the linear function of the minimum variance estimate of $\beta$.\\
	 In other words given a matrix $T$ (the linear function), the minimum variance estimate of $T\beta$  is $T\hat{\beta}=TE[\beta y^T]\big[E[yy^T]\big]^{-1}y$.
\end{theorem}
\begin{proof}
	Using the proof of the minimum variance theorem and by replacing $\beta$ with $T\beta$ we  write
\begin{align}
	E[yy^T]K^T&=E[y(T\beta)^T]=E[y\beta^TT^T]=E[y\beta^T]T^T\nonumber\\
	\Rightarrow K^T&= \left( E[yy^T]\right) ^{-1}E[y\beta^T]T^T\nonumber\\
	\Rightarrow K &= T\underbrace{E[\beta y^T](E[yy^T])^{-1}}_{\hat{\beta}}=T\hat{\beta}\quad\nonumber
\end{align}
\end{proof}

\begin{theorem}
	If $\hat{\beta}=Ky$ is the minimum variance estimate of $\beta$, then $\hat{\beta}$ is also the linear estimate that minimizes $E[(\beta-\hat{\beta})^TP(\beta-\hat{\beta})]$ for any positive semi-definite $P_{n\times n}$ . 
\end{theorem}
\begin{proof}
	Let $P^{\frac{1}{2}}$ be positive square root of $P$. According to theorem \ref{th:Tbeta}, $P^{\frac{1}{2}}\hat{\beta}$ is the minimum variance estimate of $P^{\frac{1}{2}}\beta$, which means that $P^{\frac{1}{2}}\hat{\beta}$ minimizes
	\begin{align}
	E[\|P^{\frac{1}{2}}\hat{\beta}-P^{\frac{1}{2}}\beta\|^2_2]&= E[(P^{\frac{1}{2}}\hat{\beta}-P^{\frac{1}{2}}\beta)^T(P^{\frac{1}{2}}\hat{\beta}-P^{\frac{1}{2}}\beta)]\nonumber\\
	&=E[(\hat{\beta}^T{P^{\frac{1}{2}}}^T-\beta^T {P^{\frac{1}{2}}}^T) (P^{\frac{1}{2}}\hat{\beta}-P^{\frac{1}{2}}\beta)].\nonumber
	\end{align} 
	Since $P^{\frac{1}{2}}\geq 0$ then ${{P^\frac{1}{2}}^T}=P^{\frac{1}{2}}$ and the proof follows by noting that
	\begin{align}
	E[\|P^{\frac{1}{2}}\hat{\beta}-P^{\frac{1}{2}}\beta\|^2_2]	&=E[(\hat{\beta}^T{P^{\frac{1}{2}}}-\beta^T {P^{\frac{1}{2}}}) (P^{\frac{1}{2}}\hat{\beta}-P^{\frac{1}{2}}\beta)]\nonumber\\
	&=E[\hat{\beta}^TP\hat{\beta}-\hat{\beta}^TP\beta-\beta^TP\hat{\beta}+\beta^TP\beta]=*\nonumber\\
	E[(\beta-\hat{\beta})^TP(\beta-\hat{\beta})]&=E[(\beta^T-\hat{\beta}^T)P(\beta-\hat{\beta})]
	=E[(\beta^TP-\hat{\beta}^TP)(\beta-\hat{\beta})]\nonumber\\
	&=E[\hat{\beta}^TP\hat{\beta}-\hat{\beta}^TP\beta-\beta^TP\hat{\beta}+\beta^TP\beta]=*\quad\nonumber
	\end{align}
\end{proof}

\subsection{Updating the Estimate}
We consider the problem of updating the estimate of $\beta$ if additional data becomes available.

First we define the sum of two vector subspaces  $\mathcal{Y}_1 + \mathcal{Y}_2$  of a Hilbert space $\mathcal{H}$ as consisting of all vectors in the form of $y_1+y_2$ where $y_1 \in \mathcal{Y}_1$ and $y_2 \in \mathcal{Y}_2$. We also define the vector space $\mathcal{Y}$ as the direct sum of two vector subspaces  $\mathcal{Y}=\mathcal{Y}_1 \oplus \mathcal{Y}_2$ if every vector $y \in \mathcal{Y}$ has a unique representation in the form of $y=y_1+y_2$ where $y_1 \in \mathcal{Y}_1$ and $y_2 \in \mathcal{Y}_2$. 

We know that if $\mathcal{Y}_1$ and $\mathcal{Y}_2$ are two subspaces of a Hilbert space $\mathcal{H}$, then $\mathcal{Y}_1+\mathcal{Y}_2$ is also a subspace of the space. We also know that if the subspace $\widetilde{\mathcal{Y}}_2$ is chosen such that $\widetilde{\mathcal{Y}}_2\perp \mathcal{Y}_1$ and  $\widetilde{\mathcal{Y}}_2\oplus\mathcal{Y}_1=\mathcal{Y}_1+\mathcal{Y}_2$\footnote{meaning that $\widetilde{\mathcal{Y}}_2\oplus\mathcal{Y}_1$ produces a subspace that is equal to the subspace of $\mathcal{Y}_1+\mathcal{Y}_2$}, then the projection of a vector $\beta\in\mathcal{H}$ onto $\mathcal{Y}_1+\mathcal{Y}_2$ is equal to the projection of $\beta$ onto $\mathcal{Y}_1$ plus the projection of  $\beta$ onto $\widetilde{\mathcal{Y}}_2$. This is visualized in Figure \ref{Fig:update}.



\begin{figure}[h]
	\centering
	\includegraphics[scale=0.4]{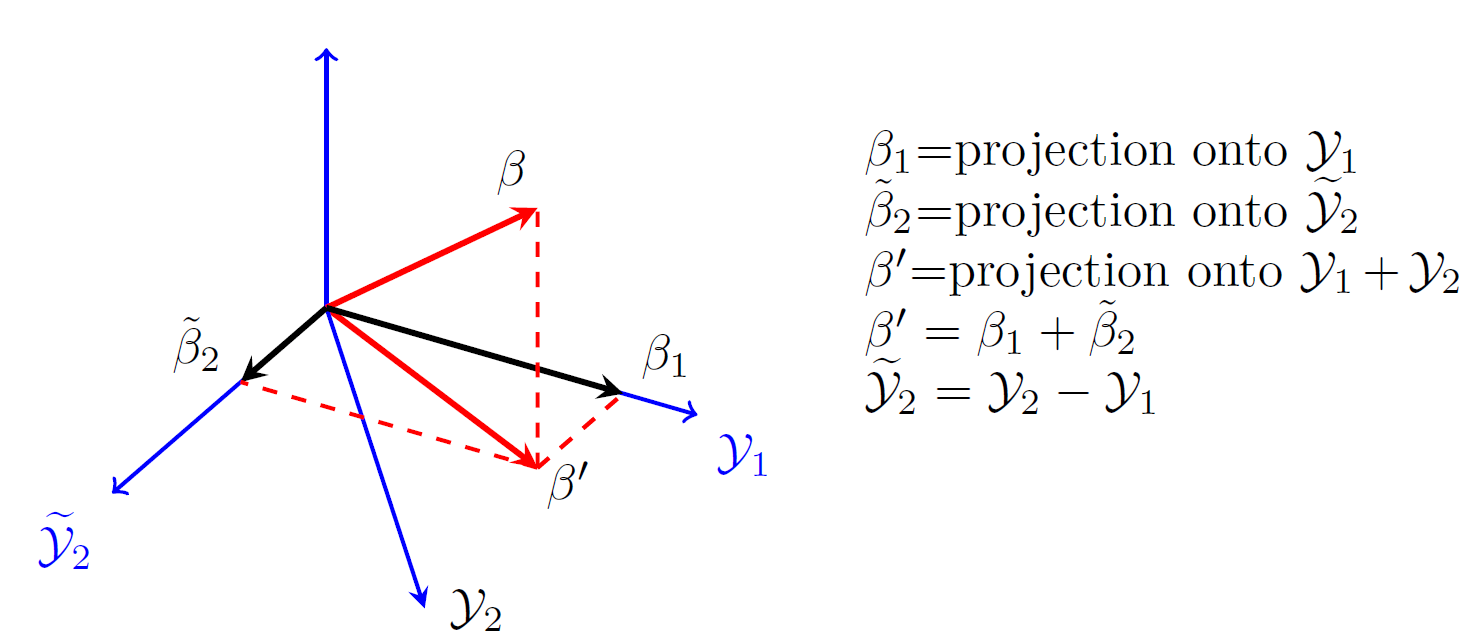} 
	\caption{A visualization of the vector $\beta$ being projected onto different subspaces.}
	\label{Fig:update}
\end{figure}

\begin{theorem}\label{thm:update}
	Let $\beta_i$ be a random variable and $\hat{\beta}_{i1}$ be the minimum variance estimate of $\beta_i$, given the random vector $y_1$. Just like the proof of the minimum variance estimator theorem, the elements of $y_1$ span a subspace $\mathcal{Y}_1\triangleq\mbox{all linear combinations of the elements of }y_1$.
	
	Let $y_2$ be a random vector and the elements of $y_2$ span a subspace $\mathcal{Y}_2$. Let $\hat{y}_2$ be the minimum variance estimate of $y_2$ in $\mathcal{Y}_1$. By the minimum variance estimate theorem, this is equivalent to saying that $\hat{y}_2$ is the orthogonal projection of the elements of ${y}_2$ onto $\mathcal{Y}_1$.
	
	Let $\tilde{y}_2=y_2-\hat{y}_2$, then the minimum variance estimate of $\beta_i$, given $y_1$ and $y_2$, is denoted by $\hat{\beta}_{i2}$ and can be found as
	\[ \hat{\beta}_{i2}=\hat{\beta}_{i1}+E[\beta_i\tilde{y}_2^T]\big[E[\tilde{y}_2\tilde{y}_2^T]\big]^{-1}\tilde{y}_2. \]
	 This is equal to saying that the orthogonal projection of $\beta_i$ onto $\mathcal{Y}_1+\mathcal{Y}_2$ is denoted by $\hat{\beta}_{i2}$.
	 In other words $\hat{\beta}_{i2}$ is $\hat{\beta}_{i1}$ plus the minimum variance estimate of $\beta_i$ given the random vector $\tilde{y}_2$. This is similar to finding the orthogonal projection of $\hat{\beta}_i$ onto $\widetilde{\mathcal{Y}}_2$ which is generated by $\tilde{y}_2$.
\end{theorem}
\begin{proof}
	The orthogonal projection of $\beta_i$ onto $\mathcal{Y}_1+\mathcal{Y}_2$ is the same as the orthogonal projection of $\beta_i$ onto $\widetilde{\mathcal{Y}}_2\oplus\mathcal{Y}_1$ since $\mathcal{Y}_1+\mathcal{Y}_2=\widetilde{\mathcal{Y}}_2\oplus\mathcal{Y}_1$, as visualized in Figure \ref{Fig:update}.
	
	Also, since $\widetilde{\mathcal{Y}}_2\perp \mathcal{Y}_1$, this orthogonal projection of $\beta_i$ onto $\widetilde{\mathcal{Y}}_2\oplus\mathcal{Y}_1$ (which we denote as $\hat{\beta}_{i2}$) is equal to the sum of individual projections onto $\mathcal{Y}_1$ (which is $\hat{\beta}_{i1}$) and onto $\widetilde{\mathcal{Y}}_2$ (which is $E[\beta_i\tilde{y}_2^T]\big[E[\tilde{y}_2\tilde{y}_2^T]\big]^{-1}\tilde{y}_2$). Therefore 
	\[ \hat{\beta}_{i2}=\hat{\beta}_{i1}+E[\beta_i\tilde{y}_2^T]\big[E[\tilde{y}_2\tilde{y}_2^T]\big]^{-1}\tilde{y}_2. \]
\end{proof}

\textbf{Intuition:} Given new data, the updating is based on the part of the new data that is orthogonal to the old data. This means that the updating is based 
on $\widetilde{\mathcal{Y}}_2$ which is orthogonal to the old data $\mathcal{Y}_1$. 

\subsubsection{Example on Updating the Estimate }
\label{sec:Example}
Suppose that an optimal estimate $\hat{\beta}$ of a random vector $\beta$ has been formed on the basis of past measurements and that
\[E[(\beta-\hat{\beta})(\beta-\hat{\beta})^T]=R. \]
 Given additional measurements of the form $y=W\beta+\varepsilon$, where $\varepsilon$ is a random vector of zero mean which is uncorrelated to both $\beta$ and the past measurements, we seek to find the updated optimal estimate $\hat{\hat{\beta}}$ and the error covariance $E[(\beta-\hat{\hat{\beta}})(\beta-\hat{\hat{\beta}})^T]$.

 Using the previous theorem we know that 
	\[ \hat{\hat{\beta}}=\hat{\beta}+E[\beta\tilde{y}^T]\big[E[\tilde{y}\tilde{y}^T]\big]^{-1}\tilde{y}, \]
	where $\tilde{y}=y-\hat{y}$ and $\hat{y}$ is the minimum variance estimate of $y$ given previous measurements.
But, the minimum variance estimate of $y$ is equal to the minimum variance estimate of $W\beta$ (since $y=W\beta$), which by Theorem \ref{th:Tbeta} is equal to $W\hat{\beta}$. Hence, we have $\tilde{y}=y-W\hat{\beta}$.

Note that $y=W\beta+\varepsilon$ and not $y=W\beta$ but since $\varepsilon$ is zero mean and uncorrelated to $\beta$ and the past measurements, the proofs of Theorems \ref{thm:Minimum Variance},\ref{th:Tbeta} and Corollary \ref{col:1} makes it clear that the minimum variance estimate of $y=W\beta+\varepsilon$ is $W\hat{\beta}$ since $E[\varepsilon\beta^T]=0$.

In order to compute $\hat{\hat{\beta}}$ we need to compute $E[\beta\tilde{y}^T]$ and $E[\tilde{y}\tilde{y}^T]$. To do this we must consider a few things 
regarding $\tilde{y}=y-W\hat{\beta}$.
\begin{itemize}
	\item $y$= the new measurement =$W\beta+\varepsilon$.
	\item $W\hat{\beta}$= the best estimate of the new measurement $y$ based on the past measurement $y_p$. 
	\item $y_p$= the past measurement from which the estimate $\hat{\beta}$ was made. ($y_p$ was used to find $\hat{\beta}$.)
	\item $y_p$ was also found from the $W\beta+\varepsilon$ process  so $y_p=W\beta+\varepsilon$.
\end{itemize}

We have previously proven that 
\[ E[(\beta-\hat{\beta})(\beta-\hat{\beta})^T]=R-RW^T(WRW^T+Q)^{-1}WR\triangleq\mathfrak{R} \]
where 
\begin{align}
Q &= E[\varepsilon\varepsilon^T]\nonumber\\
R &= E[\beta\beta^T]\nonumber\\
E[y_p\beta^T] &= WR\nonumber\\
E[\beta y_p^T] &= RW^T\nonumber\\
\hat{\beta}&=RW^T(WRW^T+Q)^{-1}y_p\nonumber\\
\hat{y}=W\hat{\beta}&=WRW^T(WRW^T+Q)^{-1}y_p.\nonumber
\end{align}

Using these previous results on the past measurement $y_p$ and by multiplying both sides of the above by $W^T$ from the right and $W$ from the left we have
\[ W\mathfrak{R}W^T=WRW^T-WRW^T(WRW^T+Q)^{-1}WRW^T. \]

We can now use the previous formulas to find $E[\beta\tilde{y}^T]$ and $E[\tilde{y}\tilde{y}^T]$ as
\begin{alignat}{2}
E[\beta\tilde{y}^T] &= E[\beta(y-W\hat{\beta})^T]&&=E[\beta(y^T-\hat{\beta}^TW^T)]\nonumber\\
&= E[\beta y^T]-E[\beta\hat{\beta}^TW^T] &&= E[\beta y^T]-E[\beta\hat{\beta}^T]W^T\nonumber\\
&=E[\beta (W\beta+\varepsilon)^T]&&-E[\beta y_p^T(WRW^T+Q)^{-1}WR]W^T\nonumber\\
&=E[\beta (\beta^T W^T+\varepsilon^T)]&&-E[\beta y_p^T](WRW^T+Q)^{-1}WRW^T\nonumber\\
&=E[\beta\beta^T] W^T+E[\beta\varepsilon^T]&&-E[\beta y_p^T](WRW^T+Q)^{-1}WRW^T\nonumber\\
&=R W^T\qquad+0&&-RW^T(WRW^T+Q)^{-1}WRW^T\nonumber\\
&=\mathfrak{R}W^T\nonumber
\end{alignat}
 and
\begin{align}
E[\tilde{y}\tilde{y}^T]&=E[(y-W\hat{\beta})(y-W\hat{\beta})^T]=E[(y-W\hat{\beta})(y^T-\hat{\beta}^TW^T)]\nonumber\\
&=E[yy^T-y\hat{\beta}^TW^T-W\hat{\beta}y+W\hat{\beta}\hat{\beta}^TW^T]\nonumber\\
&=E[yy^T]-E[y\hat{\beta}^T]W^T-WE[\hat{\beta}y^T]+WE[\hat{\beta}\hat{\beta}^T]W^T\nonumber\\
&=E[(W\beta+\varepsilon)(W\beta+\varepsilon)^T]-E[(W\beta+\varepsilon)(y_p^T(WRW^T+Q)^{-1}WR)]W^T\nonumber\\
&\quad-WE[(RW^T(WRW^T+Q)^{-1}y_p)(W\beta+\varepsilon)^T]\nonumber\\
&\quad+WE[(RW^T(WRW^T+Q)^{-1}y_p)(y_p^T(WRW^T+Q)^{-1}WR)]W^T\nonumber\\
&=WRW^T+Q-E[(W\beta+\varepsilon)y_p^T](WRW^T+Q)^{-1}WRW^T\nonumber\\
&\quad-WRW^T(WRW^T+Q)^{-1}E[y_p(W\beta+\varepsilon)^T]\nonumber\\
&\quad+WRW^T(WRW^T+Q)^{-1}E[y_py_p^T](WRW^T+Q)^{-1}WRW^T\nonumber
\end{align}

\begin{align}
&=WRW^T+Q\nonumber\\
&\quad-E[W\beta y_p^T+\varepsilon y_p^T](WRW^T+Q)^{-1}WRW^T\nonumber\\
&\quad-WRW^T(WRW^T+Q)^{-1}E[y_p\beta^TW^T+y_p\varepsilon^T]\nonumber\\
&\quad+WRW^T(WRW^T+Q)^{-1}(WRW^T+Q)(WRW^T+Q)^{-1}WRW^T\nonumber\\
&=WRW^T+Q\nonumber\\
&\quad-(E[W\beta y_p^T]+E[\varepsilon y_p^T])(WRW^T+Q)^{-1}WRW^T\nonumber\\
&\quad-WRW^T(WRW^T+Q)^{-1}(E[y_p\beta^TW^T]+E[y_p\varepsilon^T])\nonumber\\
&\quad+WRW^T(WRW^T+Q)^{-1}WRW^T\nonumber\\
&=WRW^T+Q\nonumber\\
&\quad-WE[\beta y_p^T](WRW^T+Q)^{-1}WRW^T\nonumber\\
&\quad-WRW^T(WRW^T+Q)^{-1}E[y_p\beta^T]W^T\nonumber\\
&\quad+WRW^T(WRW^T+Q)^{-1}WRW^T\nonumber\\
&\quad+WRW^T(WRW^T+Q)^{-1}WRW^T\nonumber\\
&=WRW^T+Q\nonumber\\
&\quad-WRW^T(WRW^T+Q)^{-1}WRW^T\nonumber\\
&\quad-WRW^T(WRW^T+Q)^{-1}WRW^T\nonumber\\
&\quad+WRW^T(WRW^T+Q)^{-1}WRW^T\nonumber\\
&=WRW^T-WRW^T(WRW^T+Q)^{-1}WRW^T+Q\nonumber\\
&=W\mathfrak{R}W^T+Q.\nonumber
\end{align}

Finally, noting that $\mathfrak{R}^T=\mathfrak{R}$, the   error covariance is found as
\begin{align}
&E[(\beta-\hat{\hat{\beta}})(\beta-\hat{\hat{\beta}})^T]\nonumber\\
&=E[(\beta-\hat{\beta}-\mathfrak{R}W^T(W\mathfrak{R}W^T+Q)^{-1}(y-W\hat{\beta}))(\beta-\hat{\beta}-\mathfrak{R}W^T(W\mathfrak{R}W^T+Q)^{-1}(y-W\hat{\beta}))^T]\nonumber\\
&=E[(\beta-\hat{\beta}-\mathfrak{R}W^T(W\mathfrak{R}W^T+Q)^{-1}(y-W\hat{\beta}))((\beta-\hat{\beta})^T-(y-W\hat{\beta})^T(W\mathfrak{R}W^T+Q)^{-1})W\mathfrak{R}]\nonumber\\
&=E[(\beta-\hat{\beta})(\beta-\hat{\beta})^T-(\beta-\hat{\beta})(y-W\hat{\beta})^T(W\mathfrak{R}W^T+Q)^{-1}W\mathfrak{R}\nonumber\\
&\quad-\mathfrak{R}W^T(W\mathfrak{R}W^T+Q)^{-1}(y-W\hat{\beta})(\beta-\hat{\beta})^T\nonumber\\
&\quad+\mathfrak{R}W^T(W\mathfrak{R}W^T+Q)^{-1}(y-W\hat{\beta})(y-W\hat{\beta})^T(W\mathfrak{R}W^T+Q)^{-1}W\mathfrak{R}].\nonumber
\end{align}
Noting that $E[(\beta-\hat{\beta})(\beta-\hat{\beta})^T]=\mathfrak{R}$, the error covariance is
\begin{align}
&=\mathfrak{R}-E[(\beta-\hat{\beta})(y-W\hat{\beta})^T](W\mathfrak{R}W^T+Q)^{-1}W\mathfrak{R}\nonumber\\
&\quad-\mathfrak{R}W^T(W\mathfrak{R}W^T+Q)^{-1}E[(y-W\hat{\beta})(\beta-\hat{\beta})^T]\nonumber\\
&\quad+\mathfrak{R}W^T(W\mathfrak{R}W^T+Q)^{-1}E[(y-W\hat{\beta})(y-W\hat{\beta})^T](W\mathfrak{R}W^T+Q)^{-1}W\mathfrak{R}.\nonumber
\end{align}
Since $E[(y-W\hat{\beta})(y-W\hat{\beta})^T]=W\mathfrak{R}W^T+Q$, the error covariance is
\begin{align}
&=\mathfrak{R}-E[\beta(y-W\hat{\beta})^T-\hat{\beta}(y-W\hat{\beta})^T](W\mathfrak{R}W^T+Q)^{-1}W\mathfrak{R}\nonumber\\
&\quad-\mathfrak{R}W^T(W\mathfrak{R}W^T+Q)^{-1}E[(y-W\hat{\beta})\beta^T-(y-W\hat{\beta})\hat{\beta}^T]\nonumber\\
&\quad+\mathfrak{R}W^T(W\mathfrak{R}W^T+Q)^{-1}(W\mathfrak{R}W^T+Q)(W\mathfrak{R}W^T+Q)^{-1}W\mathfrak{R}.\nonumber
\end{align}
Setting $E[\beta(y-W\hat{\beta})^T]=\mathfrak{R}W^T$ and $E[(y-W\hat{\beta}){\beta}^T]=W\mathfrak{R}$, the error covariance is
\begin{align}
&=\mathfrak{R}-\mathfrak{R}W^T(W\mathfrak{R}W^T+Q)^{-1}W\mathfrak{R}+E[\hat{\beta}(y-W\hat{\beta})^T](W\mathfrak{R}W^T+Q)^{-1}W\mathfrak{R}\nonumber\\
&\quad-\mathfrak{R}W^T(W\mathfrak{R}W^T+Q)^{-1}W\mathfrak{R}+\mathfrak{R}W^T(W\mathfrak{R}W^T+Q)^{-1}E[(y-W\hat{\beta})\hat{\beta}^T]\nonumber\\
&\quad+\mathfrak{R}W^T(W\mathfrak{R}W^T+Q)^{-1}W\mathfrak{R}.\nonumber
\end{align}
Finally, if $E[\hat{\beta}(y-W\hat{\beta})^T]=0$ then the error covariance is 
\begin{align}
E[(\beta-\hat{\hat{\beta}})(\beta-\hat{\hat{\beta}})^T]=\mathfrak{R}-\mathfrak{R}W^T(W\mathfrak{R}W^T+Q)^{-1}W\mathfrak{R}.\nonumber
\end{align}
To show that $E[\hat{\beta}(y-W\hat{\beta})^T]=0$ we write
\begin{align}
E[\hat{\beta}(y-W\hat{\beta})^T]&=E[\hat{\beta}y^T-\hat{\beta}\hat{\beta}^TW^T]=E[\hat{\beta}y^T]-E[\hat{\beta}\hat{\beta}^TW^T]\nonumber\\
&=E[\hat{\beta}(W\beta+\varepsilon)^T]-E[\hat{\beta}\hat{\beta}^T]W^T\nonumber\\
&=E[\hat{\beta}\beta^T]W^T+E[\hat{\beta}\varepsilon^T]-E[\hat{\beta}\hat{\beta}^T]W^T\nonumber\\
&=E[Ky_p\beta^T]W^T+E[Ky_p\varepsilon^T]-E[Ky_py_p^TK^T]W^T\nonumber\\
&=KE[y_p\beta^T]W^T+KE[y_p\varepsilon^T]-KE[y_py_p^T]K^TW^T.\nonumber
\end{align}
Setting $E[y_p\varepsilon^T]=0$ and $K=E[\beta y_p^T][E[y_py_p^T]]^{-1}$, the above equation is
\begin{align}
&=E[\beta y_p^T][E[y_py_p^T]]^{-1}E[y_p\beta^T]W^T-E[\beta y_p^T][E[y_py_p^T]]^{-1}E[y_py_p^T][E[y_py_p^T]]^{-1}E[y_p\beta^T ]W^T\nonumber\\
&=E[\beta y_p^T][E[y_py_p^T]]^{-1}E[y_p\beta^T]W^T-E[\beta y_p^T][E[y_py_p^T]]^{-1}E[y_p\beta^T ]W^T=0\nonumber
\end{align}

\section{Kalman Filtering}
\subsection{Dynamic Model of a Random Process}
An $n$-dimensional dynamic  random process can be modeled as follows. 
\begin{enumerate}
	\item A vector  difference equation \[x_{k+1}=\Phi_k x_k+u_k,\qquad k=0,1,2,\ldots\] which defines how the random vector $x_k$ changes. 
	
	\begin{itemize}
		\item Here $x_k$ is an $n$-dimensional state vector where each component is a random variable.
		\item $\Phi_k$ is a known $n\times n$  matrix.
		\item $u_k$ is an $n$-dimensional input random vector   with zero mean such that there is zero 
		correlation between present input at   $k$ and past input at  $l$, i.e:
		\[ E[u_k u_l^T]=Q_k\delta_{kl}=\left\lbrace\begin{array}{ll}
		Q_k&\qquad k=l\\
		0&\qquad k\neq l
		\end{array}\right. \]
		where $Q_k>0$ is a positive definite matrix.
	\end{itemize}
	\item An initial random vector $x_0$ and initial random estimate $\hat{x}_0$ with initial error covariance $ E[(x_0-\hat{x}_0)(x_0-\hat{x}_0)^T]=P_0$. 
	\item Measurements of the process in the form of 
	\[ z_k=H_k x_k+w_k,\qquad k=0,1,2,\ldots \]
	which defines how the measurements $z_k$ of the process $x_k$ are recorded.
	\begin{itemize}
		\item Here $H_k$ is a known $m\times n$ matrix.
		\item $w_k$ is an $n$-dimensional random measurement error with zero mean and 
		\[ E[w_k w_l^T]=R_k\delta_{kl}=\left\lbrace\begin{array}{ll}
		R_k&\qquad k=l\\
		0&\qquad k\neq l
		\end{array}\right. \] 
		where $R_k>0$ is a positive definite matrix.
	\end{itemize}
It is assumed that $x_0$,$u_j$,$w_k$ are all uncorrelated for $j\geq0,k\geq0$.
\end{enumerate}

\subsection{The Estimation Problem}
The estimation problem is defined as finding the  minimum-variance estimate of $x$ from measurements $z$.
We say that $\hat{x}_{k|j}$ is the optimal estimate of $x_{k}$ given $j$ measurements (or observations) of $z$. In other words,  $\hat{x}_{k|j}$ is the 
projection of $x_{k}$ onto the space $\mathcal{Z}_{j}$ generated by the random vectors $z_{0},z_{1},z_{2},\ldots,z_{j}$.

We consider the case of $k\geq j$ of either predicting future or present values given past measurements. Estimating  past values  is called the smoothing problem, which is substantially the same but with messier equations.

\subsection{Kalman Filter Theorem} \label{sec:vectorkalman}
\begin{theorem}
	The optimal estimate  $\hat{x}_{k+1|k}$ of a random state vector can be generated recursively as
	\begin{align} 
	\label{eq:1Kamlan}
	\hat{x}_{k+1|k}&=\Phi_{k}\hat{x}_{k|k-1}+\Phi_{k}P_{k}H_{k}^T[H_{k}P_{k}H^T_{k}+R_{k}]^{-1}(z_{k}-H_{k}\hat{x}_{k|k-1})
	\end{align}
	where $P_{k}$ is the $n\times n$ error covariance of $\hat{x}_{k|k-1}$ which is itself generated recursively as:
	\begin{align}
	\label{eq:2Kalman}
	P_{k+1}&=\Phi_{k}P_{k}\lbrace I-H^T_{k}[H_{k}P_{k}H^T_{k}+R_{k}]^{-1} H_{k}P_{k}\rbrace\Phi^T_{k}+Q_{k}
	\end{align}
	
	The required initial conditions are the initial estimates of $\hat{x}_{0}=\hat{x}_{0|-1}$ and its error covariance $P_{0}$.
\end{theorem}

\begin{proof}
	Suppose that $z_{0},z_{1},\ldots,z_{k - 1}$ have been measured and 
	that the estimate $\hat{x}_{k | k - 1}$ and error covariance  $P_{k} = 
	E[(\hat{x}_{k | k - 1} - x_{k})(\hat{x}_{k | k - 1} - x_{k})^T]$ have been computed. In other 
	words, we have the projection of $x_{k}$ onto the subspace $\mathcal{Z}_{k - 1}$. 
	
	At  $k$, 
	we obtain a new measurement 
	\[z_{k} = H_{k}x_{k} + w_{k} \]
	which gives us additional information about  $x_{k}$. 
	This is exactly the situation encountered in the aforementioned example of Section \ref{sec:Example}. 
	Specifically, by substituting $\hat{\hat{\beta}}=\hat{x}_{k|k}$, $ \hat{\beta}=\hat{x}_{k|k-1} $, $ \mathfrak{R}=P_{k} $, $ W=H_{k} $, $ Q=R_{k} $, $ y=z_{k} $ in    
	the previous example,
	the updated estimate of $x_{k}$ is
	\begin{align} \hat{x}_{k|k}&=\hat{x}_{k|k-1}+P_{k}H^T_{k}[H_{k}P_{k}H^T_{k}+R_{k}]^{-1}(z_{k}-H_{k}\hat{x}_{k|k-1})\nonumber 
	\end{align}
	
	with associated error covariance 
	\begin{align}
	P_{k|k}&=P_{k}-P_{k}H^T_{k}[H_{k}P_{k}H^T_{k}+R_{k}]^{-1}H_{k}P_{k}\nonumber\\
	&=P_{k}\lbrace I-H^T_{k}[H_{k}P_{k}H^T_{k}+R_{k}]^{-1} H_{k}P_{k},\rbrace\nonumber
	\end{align}
	where $\mathfrak{R}=P_{k}$ and $\mathfrak{R}-\mathfrak{R}W^T(W\mathfrak{R}W^T+Q)^{-1}W\mathfrak{R}=P_{k|k}$.
	
	Based on this optimal estimate of $x_{k}$, we can  compute the optimal 
	estimate $\hat{x}_{k + 1 | k}$ of $x_{k + 1} = \Phi_{k}\hat{x}_{k} + u_{k}$.
	We can do this  using Theorem \ref{th:Tbeta},which says that the optimal 
	estimate of $\Phi_{k}x_{k}$ is $\Phi_{k}\hat{x}_{k|k}$, and since $u_{k}$ is  uncorrelated with $z_{k}$ and $x_{k}$, the optimal estimate of $x_{k + 1}$ is 
	\begin{align}
	\hat{x}_{k+1|k}&=\Phi_{k}\hat{x}_{k|k}=\Phi_{k}\hat{x}_{k|k-1}+\Phi_{k}P_{k}H^T_{k}[H_{k}P_{k}H^T_{k}+R_{k}]^{-1}(z_{k}-H_{k}\hat{x}_{k|k-1}).\nonumber
	\end{align}
	This proves equation \ref{eq:1Kamlan}.
	
	To prove the error covariance update equation  \ref{eq:2Kalman}, we first note  that from Theorem \ref{th:Tbeta} we have
	\begin{align}
	E[(T\hat{\beta}-T\beta)(T\hat{\beta}-T\beta)^T]&=E[T(\hat{\beta}-\beta)(T(\hat{\beta}-\beta))^T]\nonumber\\
	=E[T(\hat{\beta}-\beta)(\hat{\beta}-\beta)^TT^T]&=TE[(\hat{\beta}-\beta)(\hat{\beta}-\beta)^T]T^T.\nonumber
	\end{align}
	The error covariance update $P_{k+1}$ is now
	\begin{align}
	&P_{k+1}=\nonumber\\
	&=E[(\hat{x}_{k+1|k}-x_{k+1})(\hat{x}_{k+1|k}-x_{k+1})^T]\nonumber\\
	&=E[(\Phi_{k}\hat{x}_{k|k}-(\Phi_{k}x_{k}+u_{k}))(\Phi_{k}\hat{x}_{k|k}-(\Phi_{k}x_{k}+u_{k}))^T]\nonumber\\
	&=E[(\Phi_{k}\hat{x}_{k|k}-\Phi_{k}x_{k})(\Phi_{k}\hat{x}_{k|k}-\Phi_{k}x_{k})^T]\nonumber\\
	&\quad-E[(\Phi_{k}\hat{x}_{k|k}-\Phi_{k}x_{k})u^T_{k}]-E[u_{k}(\Phi_{k}\hat{x}_{k|k}-\Phi_{k}x_{k})^T]\nonumber\\
	&\quad+E[u_{k}u^T_{k}].\nonumber
	\end{align}
	Since the error $u_{k}$ is uncorrelated with previous estimates 
	\[
	E[(\Phi_{k}\hat{x}_{k|k}-\Phi_{k}x_{k})u^T_{k}]=E[u_{k}(\Phi_{k}\hat{x}_{k|k}-\Phi_{k}x_{k})^T]=0.
	\]
	Also, we know that $E[u_{k}u^T_{k}]=Q_{k}$, therefore 
	\begin{align}
	P_{k+1}&=\Phi_{k}E[(\hat{x}_{k|k}-x_{k})(\hat{x}_{k|k}-x_{k})^T]\Phi^T_{k}+Q_{k}\nonumber\\
	&=\Phi_{k}P_{k|k}\Phi^T_{k}+ Q_{k}\nonumber\\
	&=\Phi_{k}P_{k}\lbrace I-H^T_{k}[H_{k}P_{k}H^T_{k}+R_{k}]^{-1} H_{k}P_{k}\rbrace\Phi^T_{k}+Q_{k}.\nonumber
	\end{align}

	%
	%
	
\end{proof}

\part{Bayesian Optimal Filtering }
\section{General Case}
From a Bayesian perspective, filtering means to quantify a degree of belief in the state $x_k$ at time $k$, given all the data up to time $k$ ($Z_k$) in a recursive (sequential) manner. i.e. to construct the posterior $P(x_k|Z_k)$. We do this in 2 steps:
\begin{enumerate}
	\item Prediction: uses the state model to predict the belief of state at time $k$, using $Z_{k-1}$.
	\item Update: At time $k$ when measurement $z_k$ becomes available, we will update the prediction.
\end{enumerate}
In the step of prediction we have a previous belief $P(x_{k-1}|Z_{k-1})$ and we want to know what can be predicted about $x_{k}$ i.e. we want to find $P(x_{k}|Z_{k-1})$.
We use the Chapman-Kolmograov equation:
\begin{align}
P(x_{k}|Z_{k-1})&=\frac{P(x_{k},Z_{k-1})}{P(Z_{k-1})}=\frac{\int p(x_{k},x_{k-1},Z_{k-1})dx_{k-1}}{P(Z_{k-1})}\nonumber\\
&=\frac{\int p(x_{k}|x_{k-1},Z_{k-1})p(x_{k-1}|Z_{k-1})P(Z_{k-1})dx_{k-1}}{P(Z_{k-1})}\nonumber\\
&=\int p(x_{k}|x_{k-1},Z_{k-1})p(x_{k-1}|Z_{k-1})dx_{k-1}\nonumber\\
&=\int p(x_{k}|x_{k-1})p(x_{k-1}|Z_{k-1})dx_{k-1}\nonumber
\end{align}
Assuming the state at time $k$ is only dependent on the state at time $k-1$ and is independent of the observation history $Z_{k-1}$ when $x_{k-1}$ is given. In the above equation $p(x_{k}|x_{k-1})$ is derived from the state equation. 

The step of update uses new measurement $z_{k}$ to construct the posterior $P(x_k|Z_{k})$. The update or corrector is carried out via the Bayes rule.
\begin{align}
P(x_{k}|Z_{k})&=P(x_k|z_{k},Z_{k-1})=\frac{P(x_k,z_{k}|Z_{k-1})}{P(z_{k}|Z_{k-1})}=\nonumber\\
&=\frac{P(z_{k}|x_k,Z_{k-1})P(x_k|Z_{k-1})}{P(z_{k}|Z_{k-1})}\nonumber
\end{align}
Assuming that new measurement $z_k$ is independent of the previous measurements $Z_{k-1}$ we may find the update or corrector:
\[P(x_{k}|Z_{k})=\frac{P(z_{k}|x_k)P(x_k|Z_{k-1})}{P(z_{k}|Z_{k-1})}\]
$P(z_{k}|Z_{k-1})$ can be calculated as follows:
\begin{align}
P(z_{k}|Z_{k-1})&=\int p(z_k,x_k|Z_{k-1}dx_k)=\int p(z_k|x_k,Z_{k-1})p(x_k|Z_{k-1})dx_k\nonumber\\
&=\int p(z_k|x_k)p(x_k|Z_{k-1})dx_k\nonumber.
\end{align}
In the above equation $p(z_k|x_k)$ is the likelihood function (likelihood of data $z_k$ given the state $x_k$) which can be found from the measurement equation.

Once the posterior is found, the estimate of the state can be found using the mean or mode of the posterior. 

For MMSE this estimate is defined by:
\[\hat{x}_{k|k}=E[x_k|Z_k]=\int x_kp(x_k|Z_k)dx_k\qquad (\text{MMSE})\]

For MAP the estimate is given by:
\[ \begin{array}{ccc}
\hat{x}_{k|k}= &\text{argmax} & p(x_k|Z_k) \qquad (\text{MAP})\\
& x_k & 
\end{array}\]
\begin{figure}[h]\label{Fig:probability}
	\centering
	\includegraphics[scale=0.3]{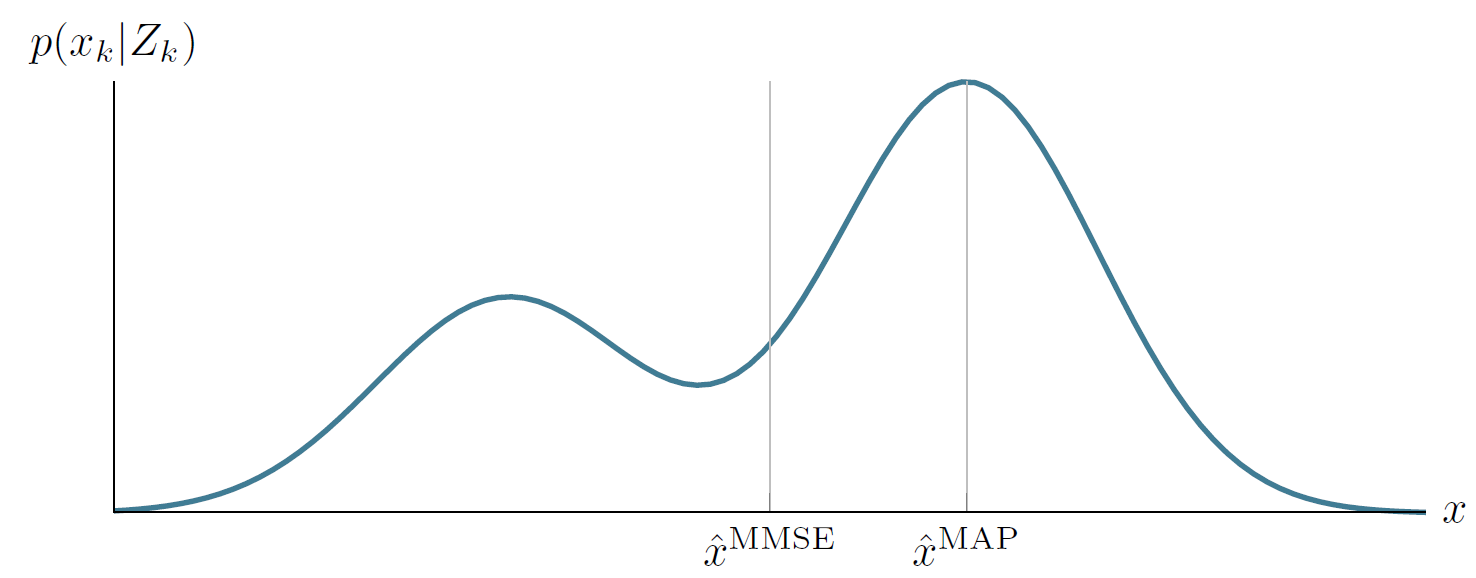}
\end{figure}

\section{Kalman Filtering}
The predictor and corrector steps can not be performed analytically for every arbitrary functions $\phi_{k-1}$, $h_k$. However it has a close-form solution for the most simple form of linear/Gaussian (Kalman filter) model.

Sequential Bayesian equations are obtained from the previous discussion 
\[\left\lbrace\begin{array}{llr}
P(x_{k}|Z_{k-1})&=\int p(x_{k}|x_{k-1})p(x_{k-1}|Z_{k-1})dx_{k-1}&\text{prediction}\\
P(x_{k}|Z_{k})&=\dfrac{P(z_{k}|x_k)P(x_k|Z_{k-1})}{\int p(z_k|x_k)p(x_k|Z_{k-1})dx_k}&\text{update/correction}
\end{array}\right.\]

\textbf{Prediction}

In the above equations $p(x_{k-1}|Z_{k-1})$ corresponds to the state $x_{k-1}$ up to time $k-1$, under the Kalman filter settings, this density turns out to be Gaussian
\[x_{k-1}|Z_{k-1}\sim\mathcal{G}(x_{k-1},\hat{x}_{k-1|k-1},P_{k-1|k-1})\]
where $\hat{x}_{k-1|k-1}$,$P_{k-1|k-1}$ can be found from previous steps and going back to $x_0$.

$p(x_k|x_{k-1})$ can be found from the state equation
\[p(x_k|x_{k-1})=\mathcal{G}(x_k,\Phi_{k-1}x_{k-1},Q_{k-1})\quad\text{or}\quad x_k|x_{k-1}\sim\mathcal{G}(x_k,\Phi_{k-1}x_{k-1},Q_{k-1})\]

Now we can put the above equations into the prediction equation
\begin{align}
P(x_{k}|Z_{k-1})&=\int \mathcal{G}(x_k,\Phi_{k-1}x_{k-1},Q_{k-1})\mathcal{G}(x_{k-1},\hat{x}_{k-1|k-1},P_{k-1|k-1})dx_{k-1}\nonumber\\
&= \mathcal{G}(x_k,\Phi_{k-1}\hat{x}_{k-1},P_{k|k-1})=\mathcal{G}(x_k,\hat{x}_{k|k-1},P_{k|k-1})\nonumber
\end{align}

Since the integrand is the multiply of 2 Gaussian PDFs, the result of the integral can be computed in form of a Gaussian PDF with
\[\left\lbrace\begin{array}{llr}
\hat{x}_{k|k-1}&=\Phi_{k-1}\hat{x}_{k|k-1}&\text{mean}\\
P_{k|k-1}&=\Phi_{k-1}P_{k-1|k-1}\Phi_{k-1}^T+Q_{k-1}&\text{covariance}
\end{array}\right.\]

Note that without new measurements to do the update(correction) step the covariance grows with time.

\textbf{Correction}

In order to find $P(x_{k}|Z_{k})$, we need to compute $P(z_{k}|x_k)$ and $P(x_k|Z_{k-1})$.  $P(z_{k}|x_k)$ can be found from the measurement equation $z_k=H_kx_k+w_k$ and $P(x_k|Z_{k-1})$ is found in the previous step. We have assumed that  $w_k\sim\mathcal{G}(w_k,0,R_k)$ therefore
\begin{align}
p(z_k|x_k)&=\mathcal{G}(z_k,H_kx_k,R_k)\nonumber\\
P(x_{k}|Z_{k-1})&=\mathcal{G}(x_k,\hat{x}_{k|k-1},P_{k|k-1})\nonumber
\end{align}

Putting the above PDFs in the corrector we obtain:
\begin{align}
P(x_{k}|Z_{k})&=\dfrac{P(z_{k}|x_k)P(x_k|Z_{k-1})}{\int p(z_k|x_k)p(x_k|Z_{k-1})dx_k}\nonumber\\
P(x_{k}|Z_{k})&=\dfrac{\mathcal{G}(z_k,H_kx_k,R_k)\mathcal{G}(x_k,\hat{x}_{k|k-1},P_{k|k-1})}{\int\mathcal{G}(z_k,H_kx_k,R_k)\mathcal{G}(x_k,\hat{x}_{k|k-1},P_{k|k-1})dx_k}\nonumber
\end{align}

In order to find the result of the above relation, first we calculate the nominator.
It can be shown that (See appendix D of \cite{Maher}). 
\[\text{det}(R_k).\text{det}(P_{k|k-1})=\text{det}(R_k+H_kP_{k|k-1}H_k)\]
and
\begin{align}
&(z_k-H_kx_k)^TR_k^{-1}(z_k-H_kx_k)+(x_k-\hat{x}_{k|k-1})^TP_{k|k-1}^{-1}(x_k-\hat{x}_{k|k-1})=\nonumber\\
&(z_k-H_k\hat{x}_{k|k-1})^T(R_k+H_kP_{k|k-1}H_k)^{-1}(z_k-H_k\hat{x}_{k|k-1})+\nonumber\\
&(x_k-\hat{x}_{k|k})^T(P_{k|k-1}+H_k^TR_k^{-1}H_k)(x_k-\hat{x}_{k|k})\nonumber
\end{align}
Then we write
\begin{align}
&\mathcal{G}(z_k,H_kx_k,R_k)\mathcal{G}(x_k,\hat{x}_{k|k-1},P_{k|k-1})=\qquad\qquad\qquad\qquad\qquad\qquad\nonumber\\
&\frac{1}{\sqrt{\text{det}(2\pi R_k)}}e^{-\frac{1}{2}(z_k-H_kx_k)^TR_k^{-1}(z_k-H_kx_k)}\frac{1}{\sqrt{\text{det}(2\pi P_{k|k-1})}}e^{-\frac{1}{2}(x_k-\hat{x}_{k|k-1})^TP_{k|k-1}^{-1}(x_k-\hat{x}_{k|k-1})}=\nonumber\\
&\frac{1}{2\pi\sqrt{ \text{det}(R_k).\text{det}(P_{k|k-1})} }e^{-\frac{1}{2}\left[(z_k-H_kx_k)^TR_k^{-1}(z_k-H_kx_k)+(x_k-\hat{x}_{k|k-1})^TP_{k|k-1}^{-1}(x_k-\hat{x}_{k|k-1})\right]}=\nonumber\\
&\frac{1}{\sqrt{\text{det}(2\pi(R_k+H_kP_{k|k-1}H_k))} }e^{-\frac{1}{2}(z_k-H_k\hat{x}_{k|k-1})^T(R_k+H_kP_{k|k-1}H_k)^{-1}(z_k-H_k\hat{x}_{k|k-1})}\nonumber\\
&\times \frac{1}{\sqrt{\text{det}(2\pi(P_{k|k-1}+H_k^TR_k^{-1}H_k)^{-1}})}e^{-\frac{1}{2}(x_k-\hat{x}_{k|k})^T(P_{k|k-1}+H_k^TR_k^{-1}H_k)(x_k-\hat{x}_{k|k})}\nonumber
\end{align}
Thus the nominator is
\begin{align}
&\mathcal{G}(z_k,H_kx_k,R_k)\mathcal{G}(x_k,\hat{x}_{k|k-1},P_{k|k-1})=\nonumber\\
&\mathcal{G}(z_k,H_k\hat{x}_{k|k-1},R_k+H_kP_{k|k-1}H_k^T)\mathcal{G}(x_k,\hat{x}_{k|k-1},(P_{k|k-1}+H_k^TR_k^{-1}H_k)^{-1})\nonumber
\end{align}
Secondly integrating the above relationship over $x_k$ gives the denominator
\begin{align}
&\int\mathcal{G}(z_k,H_kx_k,R_k)\mathcal{G}(x_k,\hat{x}_{k|k-1},P_{k|k-1})dx_k=\nonumber\\
&\mathcal{G}(z_k,H_k\hat{x}_{k|k-1},R_k+H_kP_{k|k-1}H_k^T)\underbrace{\int\mathcal{G}(x_k,\hat{x}_{k|k-1},(P_{k|k-1}+H_k^TR_k^{-1}H_k)^{-1})dx_k}_{1}=\nonumber\\
&\mathcal{G}(z_k,H_k\hat{x}_{k|k-1},R_k+H_kP_{k|k-1}H_k^T)\nonumber
\end{align}
substituting in the main formula gives the updated PDF. Here the covariance of the updated PDF is defined as $P_{k|k}^{-1}\triangleq P_{k|k-1}^{-1}+H_k^TR_k^{-1}H_k$
\begin{align}
P(x_{k}|Z_{k})&=\dfrac{\mathcal{G}(z_k,H_k\hat{x}_{k|k-1},R_k+H_kP_{k|k-1}H_k^T)\mathcal{G}(x_k,\hat{x}_{k|k},P_{k|k})}{\mathcal{G}(z_k,H_k\hat{x}_{k|k-1},R_k+H_kP_{k|k-1}H_k^T)\underbrace{\int\mathcal{G}(x_k,\hat{x}_{k|k-1},(P_{k|k-1}+H_k^TR_k^{-1}H_k)^{-1})dx_k}_{1}}\nonumber\\
P(x_{k}|Z_{k})&=\mathcal{G}(x_k,\hat{x}_{k|k},P_{k|k})\nonumber
\end{align}
Finally we must find $\hat{x}_{k|k}$ which is the updated estimate.
By definition we have 
\[P_{k|k}^{-1}\hat{x}_{k|k}=P_{k|k-1}^{-1}\hat{x}_{k|k-1}+H_k^TR_k^{-1}z_k\]
To obtain  $\hat{x}_{k|k}$ we find $P_{k|k}$ using the matrix inversion lemma\footnote{$(A+UCV)^{-1}=A^{-1}-A^{-1}U(C^{-1}+VA^{-1}U)^{-1}VA^{-1}$}
\begin{align}
P_{k|k}&={\left[P_{k|k}^{-1}\right]}^{-1}=(P_{k|k-1}^{-1}+H_k^TR_k^{-1}H_k)^{-1}\nonumber\\
&=P_{k|k-1}-P_{k|k-1}H_k^T(R_k+H_kP_{k|k-1}H_k^T)^{-1}H_kP_{k|k-1}\nonumber\\
&=\left(I-P_{k|k-1}H_k^T(R_k+H_kP_{k|k-1}H_k^T)^{-1}H_k\right)P_{k|k-1}\nonumber\\
&=(I-K_kH_k)P_{k|k-1},\quad K_k\triangleq P_{k|k-1}H_k^T(R_k+H_kP_{k|k-1}H_k^T)^{-1}\nonumber
\end{align}
Then by multiplying $P_{k|k}$ and $P_{k|k}^{-1}\hat{x}_{k|k}$ we get
\begin{align}
P_{k|k}P_{k|k}^{-1}\hat{x}_{k|k}&=(I-K_kH_k)P_{k|k-1}\left((P_{k|k-1}^{-1}\hat{x}_{k|k-1}+H_k^TR_k^{-1}z_k)\right)\nonumber\\
\hat{x}_{k|k}&=(I-K_kH_k)\hat{x}_{k|k-1}+(I-K_kH_k)P_{k|k-1}H_k^TR_k^{-1}z_k\nonumber\\
&=\hat{x}_{k|k-1}-K_kH_k\hat{x}_{k|k-1}+P_{k|k-1}H_k^TR_k^{-1}z_k-K_kH_kP_{k|k-1}H_k^TR_k^{-1}z_k\nonumber\\
&=\hat{x}_{k|k-1}+(P_{k|k-1}H_k^T(R_k+H_kP_{k|k-1}H_k^T)^{-1}(R_k+H_kP_{k|k-1}H_k^T)R_k^{-1}\nonumber\\
&-K_kH_kP_{k|k-1}H_k^TR_k^{-1})z_k-K_kH_k\hat{x}_{k|k-1}\nonumber\\
&=\hat{x}_{k|k-1}+(K_k(I+H_kP_{k|k-1}H_k^TR_k^{-1})-K_kH_kP_{k|k-1}H_k^TR_k^{-1})z_k-K_kH_k\hat{x}_{k|k-1}\nonumber\\
&=\hat{x}_{k|k-1}+(K_k+K_kH_kP_{k|k-1}H_k^TR_k^{-1}-K_kH_kP_{k|k-1}H_k^TR_k^{-1})z_k-K_kH_k\hat{x}_{k|k-1}\nonumber\\
&=\hat{x}_{k|k-1}+K_k(z_k-H_k\hat{x}_{k|k-1})\nonumber
\end{align}
\subsection*{Summary}
In summery the predictor equations are:
\[\left\lbrace\begin{array}{llr}
\hat{x}_{k|k-1}&=\Phi_{k-1}\hat{x}_{k|k-1}&\text{mean}\\
P_{k|k-1}&=\Phi_{k-1}P_{k-1|k-1}\Phi_{k-1}^T+Q_{k-1}&\text{covariance}
\end{array}\right.\]
And the update equations 
\[\left\lbrace\begin{array}{llr}
\hat{x}_{k|k}&=\hat{x}_{k|k-1}+K_k(z_k-H_k\hat{x}_{k|k-1})\quad\\
P_{k|k}&=(I-K_kH_k)P_{k|k-1} \\
K_k&= P_{k|k-1}H_k^T(R_k+H_kP_{k|k-1}H_k^T)^{-1}\quad\text{Kalman gain}
\end{array}\right.\]

In order to obtain the results of section \ref{sec:vectorkalman} we put the update equations in the predictor to find the prediction of the next state\footnote{In section \ref{sec:vectorkalman} the notation $P_{k+1}$ represents $P_{k+1|k}$ and is used for simplicity.}
\[\left\lbrace\begin{array}{llr}
\hat{x}_{k+1|k}&=\Phi_{k}\hat{x}_{k|k-1}+\Phi_{k}P_{k|k-1}H_{k}[H_{k}P_{k|k-1}H'_{k}+R_{k}]^{-1}(z_{k}-H_{k}\hat{x}_{k|k-1})\\
P_{k+1|k}&=\Phi_{k}P_{k}\lbrace I-H'_{k}[H_{k}P_{k|k-1}H'_{k}+R_{k}]^{-1}H_{k}P_{k|k-1}\rbrace\Phi'_{k}+Q_{k}
\end{array}\right.\]


\begin{thebibliography}{}
\bibitem{Vspace}
Luenberger D.: Optimization by Vector Space Methods. Chapter 4, Wiley,  1969.
\bibitem{paoulis}
A. Papoulis and S. U. Pillai.: Probability, random variables, and stochastic processes. McGraw-Hill Education, 2002.
\bibitem{OptimalF}
 Anderson B. D. O. and  Moore J. B.: Optimal Filtering. New York: Dover, 2005.
\bibitem{nonlin}
Chen, Zhe.: Bayesian filtering: From Kalman filters to particle filters, and beyond. Statistics 182.1, 1-69, 2003.
\bibitem{kalman}
R. E. Kalman.: A new approach to linear filtering and prediction problem, Trans. ASME, Ser. D, J. Basic Eng., vol. 82, pp. 34–45, 1960. 
\bibitem{Kalman-bucy}
 R. E. Kalman and R. S. Bucy: New results in linear filtering and prediction theory, Trans. ASME, Ser. D, J. Basic Eng., vol. 83, pp. 95–107, 1961. 
\bibitem{Bayesian}
 Y. C. Ho and R. C. K. Lee, “A Bayesian approach to problems in stochastic estimation and control,” IEEE Trans. Automat. Contr., vol. 9, pp. 333–339, Oct. 1964.
 \bibitem{Bayesian2} 
 Peterka, V. "Bayesian approach to system identification." Trends and Progress in System identification 1 (1981): 239-304.
\bibitem{innov}
 K. Thomas.: The innovations approach to detection and estimation theory, Proc. IEEE, vol. 58, pp. 680–695, 1970.
\bibitem{Faragher}
 Faragher R.: Understanding the Basis of the Kalman Filter Via a Simple and Intuitive Derivation . IEEE Signal Processing Magazine, 2012
\bibitem{LAlgebra}
Meyer C. D.: Matrix Analysis and Applied Linear Algebra.
\bibitem{Maher}
Mahler, Ronald PS. Statistical multisource-multitarget information fusion. Artech House, Inc., 2007.
\end{thebibliography}
\end{document}